\documentclass[conference]{IEEEtran}

\usepackage[absolute]{textpos}
\newcommand{\copyrightstatement}{
\begin{textblock}{0.84}(0.08,0.95)    
\centering 978-1-5386-2098-4/17/\$31.00 \copyright 2017 IEEE.
\end{textblock}
}

\usepackage{graphics}
\usepackage{epstopdf}
\usepackage[pdftex]{graphicx}
\usepackage{stfloats}
\usepackage{array}
\usepackage{float}
\usepackage[cmex10]{amsmath}
\usepackage{amsmath,amsfonts,amssymb}
\usepackage{graphicx,graphics}
\usepackage{enumerate}
\usepackage{color}
\usepackage{verbatim}
\usepackage{amsthm}
\usepackage{multirow}
\usepackage{subcaption}
\usepackage{siunitx}
\usepackage{amsmath,amsthm}
\usepackage{cleveref}
\usepackage[font=footnotesize]{caption}
\usepackage{fixltx2e}
\usepackage[USenglish]{babel}

\newcommand{\e}{\mathrm{e}}

\newcommand{\erf}{\mathrm{erf}}
\begin{document}

\copyrightstatement

\title{ Performance of Analog Nonlinear Filtering for Impulsive Noise Mitigation in OFDM-based PLC Systems}

\author{
\author{Reza Barazideh$^{\dag}$, Balasubramaniam Natarajan$^{\dag}$, Alexei V. Nikitin$^{\dag,\ast}$, Ruslan L. Davidchack $^{\ddag}$ \\
\small $^{\dag}$ Department of Electrical and Computer Engineering, Kansas State University, Manhattan, KS, USA.\\
$^{\ast}$ Nonlinear Corp., Wamego, KS 66547, USA.\\
$^{\ddag}$  Dept. of Mathematics, U. of Leicester, Leicester, LE1 7RH, UK.\\
Email:\{rezabarazideh,bala\}@ksu.edu, avn@nonlinearcorp.com, rld8@leicester.ac.uk}
}

\maketitle

\begin{abstract}
Asynchronous and cyclostationary impulsive noise can severely impact the bit-error-rate (BER) of OFDM-based powerline communication systems. In this paper, we analyze an adaptive nonlinear analog front end filter that mitigates various types of impulsive noise without detrimental effects such as self-interference and out-of-band power leakage caused by other nonlinear approaches like clipping and blanking. Our proposed Adaptive Nonlinear Differential Limiter (ANDL) is constructed from a linear analog filter by applying a feedback-based nonlinearity, controlled by a single resolution parameter. We present a simple practical method to find the value of this resolution parameter that ensures the mitigation of impulsive without impacting the desired OFDM signal. Unlike many prior approaches for impulsive noise mitigation that assume a statistical noise model, ANDL is blind to the exact nature of the noise distribution, and is designed to be fully compatible with existing linear front end filters. We demonstrate the potency of ANDL by simulating the OFDM-based narrowband PLC compliant with the IEEE standards. We show that the proposed ANDL outperforms other approaches in reducing the BER in impulsive noise environments.

\end{abstract}


\section{Introduction}
Smart Grid is a concept that enables wide-area monitoring, two-way communications, and fault detection in power grids, by exploiting multiple types of communications technologies, ranging from wireless to wireline \cite{Galli2011ForTheGrid}. Thanks to the ubiquitousness of powerline infrastructure, low deployment costs, and its wide frequency band, powerline communication (PLC) has become a choice for a variety of smart grid applications~\cite{Lin13impulsive_SparseBayesian}.
In particular, there has been increasing demand in developing narrowband PLC (NB-PLC) systems in the 3-–500 kHz band, offering data rates up to 800 kbps \cite{Galli2011ForTheGrid,NikitinISPLC15}. In order to achieve such a data rates, multicarrier modulation techniques such as orthogonal frequency division multiplexing (OFDM) are preferred due to their robust performance in frequency-selective channels \cite{Zhidkovn08_Simpleanalysis}.
Since the powerline infrastructure is originally designed for power delivery and not for data communications \cite{NikitinISPLC15}, OFDM-based PLC solutions face many challenges such as noise, impedance mismatching and attenuation. Powerline noise typically generated by electrical devices connected to the powerlines and coupled to the grid via conduction and radiation is a major issue in PLC \cite{lin2013non}. Due to its technogenic (man-made) nature, this noise is typically non-Gaussian and impulsive, as has been verified by field measurements. Therefore, PLC noise can be modelled as combination of two terms: thermal noise which is assumed to be additive white Gaussian noise (AWGN), and the impulsive noise that may be synchronous or asynchronous relative to the main frequency \cite{Zimmermann02analysis}. It is observed that the primary noise component in broadband PLC (BB-PLC) \cite{Zimmermann2000analysis,nassar2011} is asynchronous and impulsive with short duration, i.e., high power impulses (up to 50 dB above thermal noise power \cite{Zimmermann2000analysis}) with random arrivals. In \cite{Nassar12cyclostationary} and IEEE P1901.2 standard \cite{Standard}, it is shown that in NB-PLC, the dominant non-Gaussian noise is a “quasi-periodic” impulsive noise ({\em cyclostationary noise}). Such noise occurs periodically with half the AC (Alternating Current) cycle with the duration ranging from hundreds of microseconds to a few milliseconds. However, it has been also claimed that asynchronous impulsive noise is simultaneously present in the higher frequency bands of NB-PLC \cite{Lin13impulsive_SparseBayesian}, \cite{NikitinISPLC15}.

The reduction in sub-channel signal-to-noise ratio (SNR) in highly impulsive noise environments such as PLC can be too severe to handle by forward error correction (FEC) and frequency-domain block interleaving (FDI) \cite{Nassar12local}, or time-domain block interleaving (TDI) \cite{Time_Interleaving}. Various approaches to deal with impulsive noise in OFDM have been proposed in prior works. Many of those approaches assume a statistical model of the impulsive noise and use parametric methods in the receiver to mitigate impulsive noise. Considering a specific statistical noise model, one can design a periodically switching moving average noise whitening filter \cite{Lin12cyclostationary}, linear minimum mean square error (MMSE) equalizer in frequency domain \cite{Yoo08asymptotic} or iterative decoder \cite{Haring-2003-Iterative-decoding} to mitigate cyclostationary noise. Such parametric methods require the overhead of training and parameter estimation. In addition, difficulty in parameter estimation and model mismatch degrade the system performance in time varying non stationary noise.

Alternately, nonlinear approaches can be implemented in order to suppress the effect of impulsive noise. The performance of memoryless digital nonlinear methods such as clipping \cite{Tseng-2012-robust-clipping}, blanking \cite{Blanking}, and combined blanking-clipping \cite{Blanking-Clipping} have been investigated in prior literature. It has been shown that for these methods, good performance is achieved only for asynchronous impulsive noise, and for high signal-to-interference ratios (SIR) \cite{Zhidkovn08_Simpleanalysis}. To address the challenge of severe impulsive noise conditions, a two-stage nulling algorithm based on iterative channel estimation is proposed in \cite{Two_Stage_Iterative}.
However, all these digital nonlinear approaches are implemented after the analog-to-digital convertor (ADC). The main drawback of these approaches lies in the fact that during the process of analog-to-digital conversion, the signal bandwidth is reduced and an initially impulsive broadband noise will appear less impulsive \cite{Nikitin-2011b-EURASIP}-\cite{Nikitin2015}. This makes the removal of impulsivity much harder by digital filters. Although, such problems can be overcome by increasing the sampling rate, it increases complexity and cost making it inefficient for real-time implementation \cite{NikitinISPLC15}, \cite{Nikitin04adaptive_rank-filter}.

In this work, unlike other prior approaches we mitigate impulsive noise in the analog domain before the ADC by using a blind adaptive analog nonlinear filter, referred to as Adaptive Nonlinear Differential Limiter (ANDL). In this technique, the adaptation is done by adjusting a single resolution parameter to work efficiently in the presence of various types of impulsive noise (asynchronous and cyclostationary impulsive noise, or combination of both) without the detailed knowledge of the noise distribution. Since ANDL is nonlinear, their effects on the desired signal are totally different than on the impulsive noise. This feature allows the filter to increase the signal to noise ratio (SNR) in the desired bandwidth by reducing the spectral density of non-Gaussian noise without significantly affecting the desired signal. Analog structure of this method allows us to use ANDL either as a stand-alone approach, or in combination with other digital impulsive noise reduction approaches. Our preliminary work in \cite{NikitinISPLC15} highlighted the basics of the ANDL approach, and its results were limited to the study of general behavior of SNR in a conceptual system without realistic OFDM transmitter and receiver modules. In this paper, we extend the analysis by explicitly qualifying the bit error rate (BER) performance of a practical OFDM-based PLC system. Additionally, unlike \cite{NikitinISPLC15}, we illustrate the performance gains offered by ANDL relative to other conventional approaches such as blanking and linear filtering. Finally, for the first time, we present a simple method to determine an effective value for the resolution parameter that maximizes signal quality while mitigating the impulsive noise.


The remainder of this paper is organized as follows. Section \ref{sec:System and Noise Models} describes the considered system and noise models. Proposed ANDL approach including resolution parameter calculation is described in section III. Section IV presents simulation results and finally conclusions are drawn in Section V.

\section{System and Noise Models}\label{sec:System and Noise Models}

We consider an OFDM system with complex baseband equivalent representation shown in Fig. \ref{fig:System Model}. In this system, information bits are independently and uniformly generated and mapped into baseband symbols $s_k$ based on phase shift keying (PSK) or quadrature amplitude modulation (QAM) scheme with Gray coding. The symbols $s_k$ are sent through an OFDM modulator which employs an inverse discrete Fourier transform (IDFT) to transmit the symbols over orthogonal subcarriers. The output analog signal envelope in time domain can be written as

\begin{figure*}
\centering
\includegraphics[scale=.44]{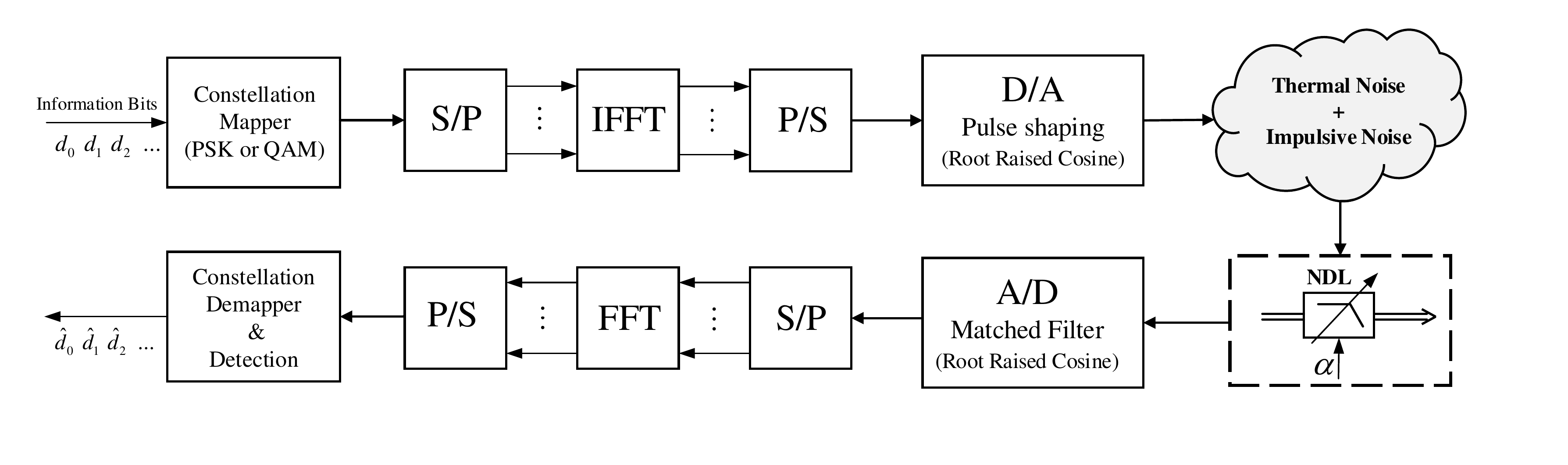}
\caption{System Model}
\label{fig:System Model}
\end{figure*}
\begin{equation}
s(t) = \frac{1}{{\sqrt N }}\sum\limits_{k = 0}^{N - 1} {{s_k}}\,{\e^{j\frac{{2\pi kt}}{{{T}}}}}p(t),\,\,\,\,\,0 < t < {T},
\end{equation}
where $N$ is the number of subcarriers, $T$ is duration of one OFDM symbol and $p(t)$ denotes the root-raised-cosine pulse shape with roll-off factor $0.25$. It is assumed that the number of subcarriers is large enough so that Central Limit Theorem (CLT) can be invoked to show that the real and imaginary parts of the OFDM signal $s(t)$ can be modeled as Gaussian random variables. In general, for different applications, we can construct an OFDM symbol with $M$ non-data subcarriers and $N-M$ data subcariers. The non-data subcarriers are either pilots for channel estimation and synchronization, or nulled for spectral shaping and inter-carrier interference reduction. Without loss of generality, the power of transmitted signal is normalized to unity, i.e., $\sigma _s^2 = 1$. Since the primary focus of this work is to study the impact of impulsive noise on OFDM performance, we consider a simple additive noise channel model where the received signal corresponds to
\begin{equation} \label{recived signal}
r(t) = s(t) + w(t) + i(t).
\end{equation}
Here, $s(t)$ denotes the desired signal with variance $\sigma _s^2$, $w(t)$ is complex Gaussian noise with mean zero and variance $\sigma _w^2$, and $i(t)$ represents the impulsive noise which is not Gaussian. The receiver involves a typical OFDM demodulator as shown in Fig. \ref{fig:System Model}. This traditional receiver structure is modified in order to deal with impulsive noise $i(t)$ as briefly discussed in the introduction section. Unlike most conventional impulsive noise mitigation approaches which are applied after the ADC, the proposed ANDL is implemented before the ADC. In the following, we begin with a review of the impulse noise models commonly encountered in PLC systems.

\subsection{Impulsive Noise Models}
Two types of impulsive noise that are dominant in the 3--500 KHz band for NB-PLC and in the 1.8--250 MHz band for BB-PLC are cyclostationary impulsive noise, and asynchronous impulsive noise, respectively \cite{Lin13impulsive_SparseBayesian}. Since both types of impulsive noises are presented in the NB-PLC \cite{Lin13impulsive_SparseBayesian}, \cite{NikitinISPLC15}, our impulsive noise model consists of both cyclostationary and asynchronous impulsive noises.

\subsubsection{Cyclostationary impulsive noise}

This type of impulsive noise has a duration ranging from hundreds of microseconds to a few milliseconds \cite{Lin13impulsive_SparseBayesian}, \cite{NikitinISPLC15}. Based on field measurements \cite{Standard}, the dominant part of this noise is a strong and narrow exponentially decaying noise burst that occurs periodically with half the AC cycle ($f_{\rm AC}=60Hz$). Therefore, we can model such noise as
\begin{equation}\label{CS}
i_{\rm cs}(t) = A_{\rm cs}\, \nu (t) \sum\limits_{k = 1}^\infty  \exp \left(\! \frac{ - t \!+\! \frac{k}{2f_{\rm AC}}}{\tau _{\rm cs}} \!\!\right) \theta \left(t \!-\! \frac{k}{2f_{\rm AC}} \right),
\end{equation}
where $A_{cs}$ is a constant, $\tau_{cs}$ is decaying time parameter, $\nu(t)$ is complex white Gaussian noise process with zero mean and variance one, and $\theta (t)$ is Heaviside step function. The spectral density of this noise is shaped based on measured spectrum of impulsivity in practice (power spectrum density (PSD) decaying at an approximate rate of 30 dB per 1 MHz) \cite{Standard}. The resulting time domain and frequency domain representation of this noise is depicted in Fig. \ref{fig:Cyclostationary}.

\subsubsection{Asynchronous impulsive noise}

This type of impulsive noise consists of short duration and high power impulses with random arrival. Mathematically, we have
\begin{equation}\label{AS}
i_{\rm as}(t) = \nu (t) {\sum\limits_{k = 1}^\infty  {A_k}\, \theta (t - {t_k})\, \e^{\frac{{ - t + {t_k}}}{\tau _{\rm as}}}}\,,
\end{equation}
where $A_k$ is the amplitude of ${k^{th}}$ pulse, $t_k$ is a arrival time of a poisson process with parameter $\lambda$, and  $\tau_{as}$ is decaying time parameter and has a duration about few microseconds. The time domain and frequency domain representation of this noise is depicted in Fig. \ref{fig:Asynchronous impulsive noise}.

\begin{figure}[t]
\centering
\includegraphics[width=.5\textwidth,height=50mm]{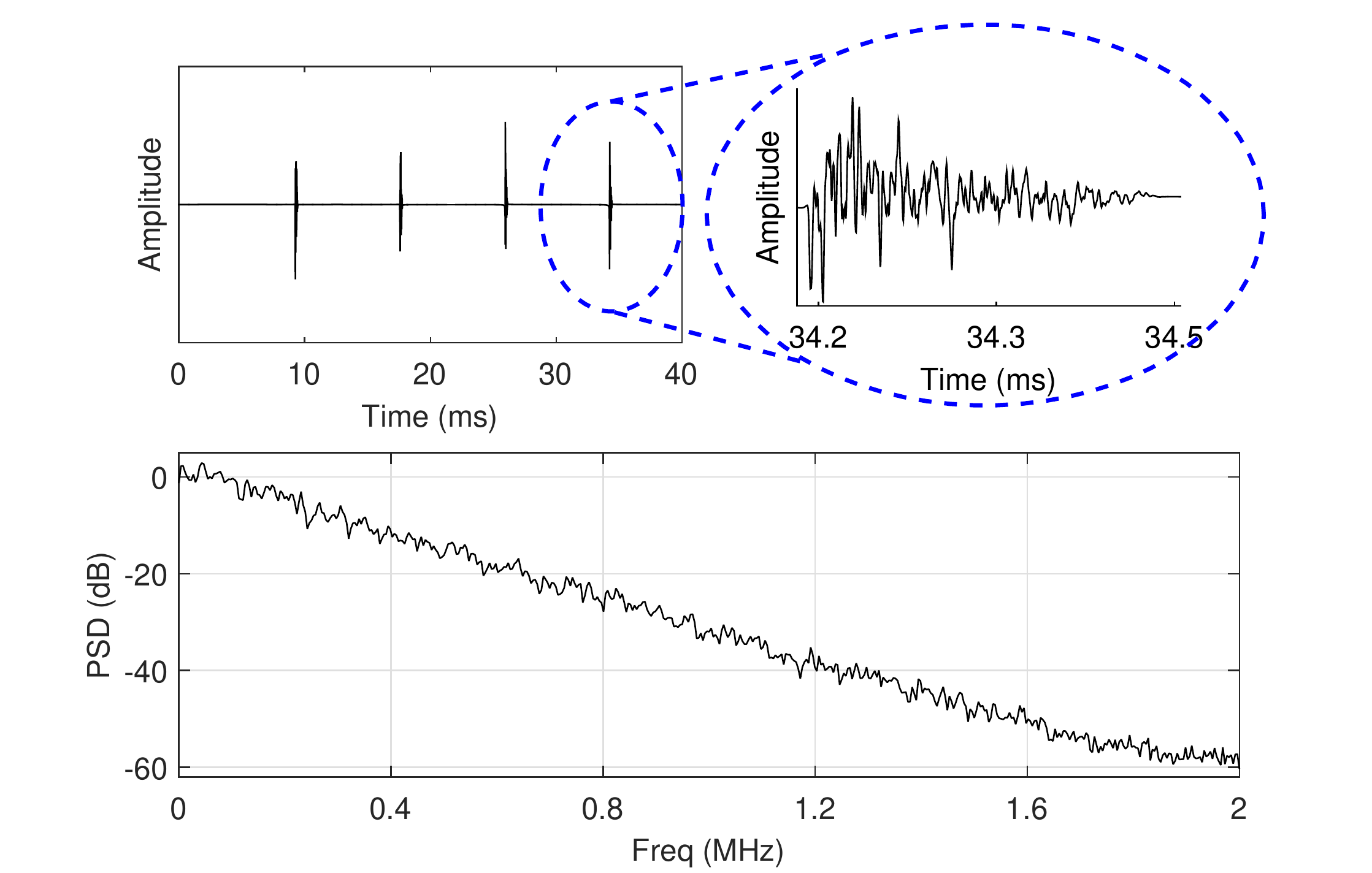}
\caption{Cyclostationary impulsive noise}
\label{fig:Cyclostationary}
\end{figure}

\begin{figure}[t]
\centering
\includegraphics[width=.5\textwidth,height=50mm]{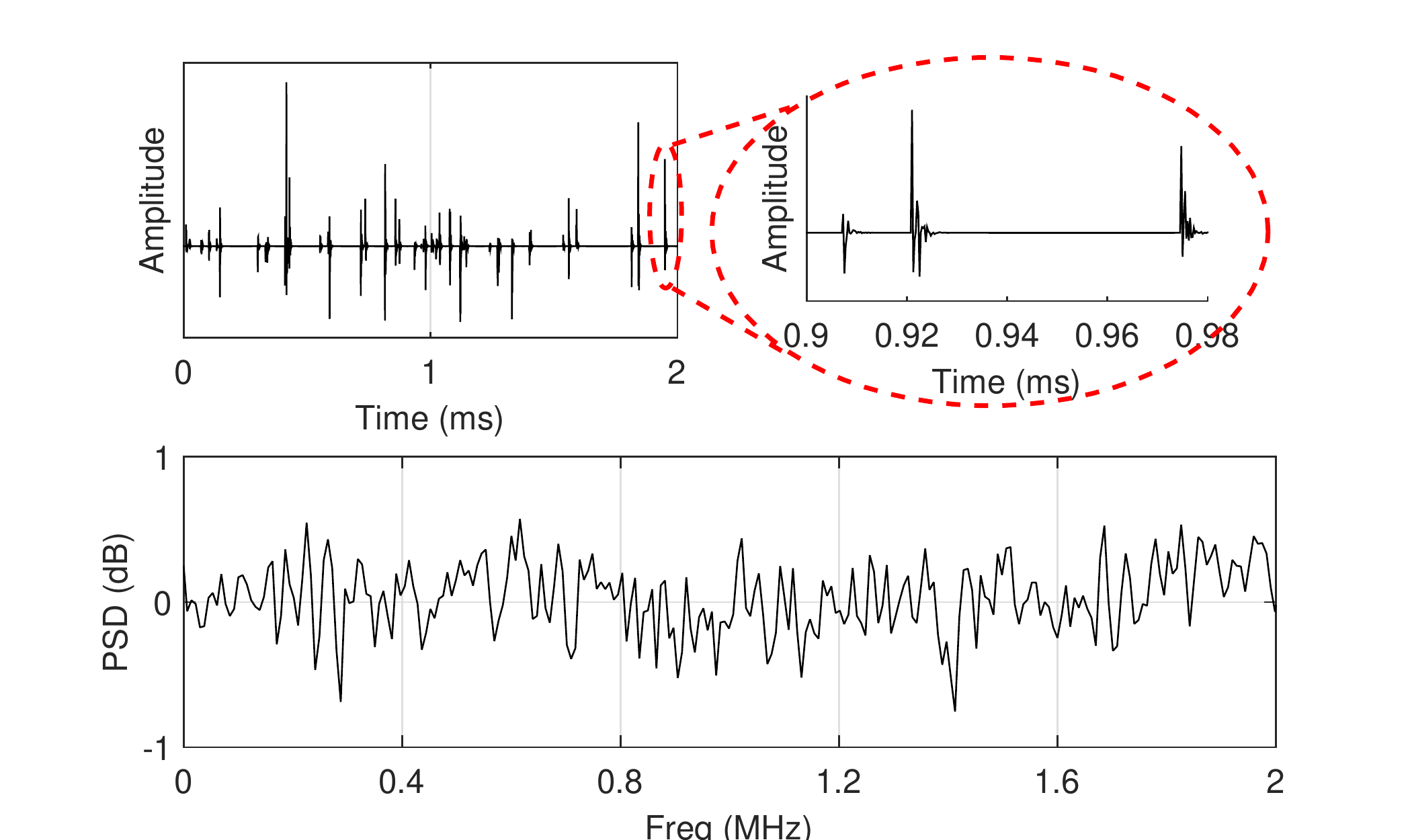}
\caption{Asynchronous impulsive noise}
\label{fig:Asynchronous impulsive noise}
\end{figure}

\section{ANDL Design}\label{sec: NDL Basics}

In this section, we provide an introduction to the basics of the ANDL and the method that can be used to find an effective value for the resolution parameter of the filter to mitigate impulsive noise.

\subsection{ANDL Formulation}


ANDL is a blind adaptive analog nonlinear filter that can be perceived as a 1st order time varying linear filter with the time parameter $\tau(t)$, that depends on the magnitude of the difference between the input and the output, as discussed in our previous work \cite{NikitinISPLC15,nikitin2014method}. Thus, we have

\begin{equation} \label{eq:1st order CDL}
  \chi(t) = x(t) - \tau(|x(t)-\chi(t)|)\, \dot{\chi}(t)\,,
\end{equation}
 where $x(t)$ and $\chi(t)$ are the input and output of the filter, respectively, and the dot denotes the first time derivative. As illustrated in Fig.~\ref{fig:Tau}, the time parameter~${\tau(t)= \tau(|x(t)-\chi(t)|)}$ is given by
\begin{equation} \label{eq:CDL tau}
  \tau(|x(t)-\chi(t)|)  = \tau_0 \times \left\{
  \begin{array}{cc}
    \!\! 1 &|x(t)-\chi(t)| \le \alpha(t)\\
    \!\!\frac{|x(t)-\chi(t)|}{\alpha(t)} & \mbox{otherwise}
  \end{array}\right.,
\end{equation}
where $\tau_0$ is a fixed time constant that ensures the desired bandwidth and $\alpha(t)$ is the resolution parameter of the filter and should be determined to mitigate the impulsive noise efficiently. Although in general the ANDL is a nonlinear filter, it behaves like a linear filter as long as there are no outliers and the magnitude of the difference signal $\left| {x(t) - \chi(t)} \right|$ remains within a certain range determined by the resolution parameter. However, when outliers are encountered, the proper selection of resolution parameter ensures that the magnitude of the corresponding outliers are suppressed by the nonlinear response of the ANDL.


\noindent
\begin{figure}[t]
\centering
\includegraphics[width=.45\textwidth,height=35mm]{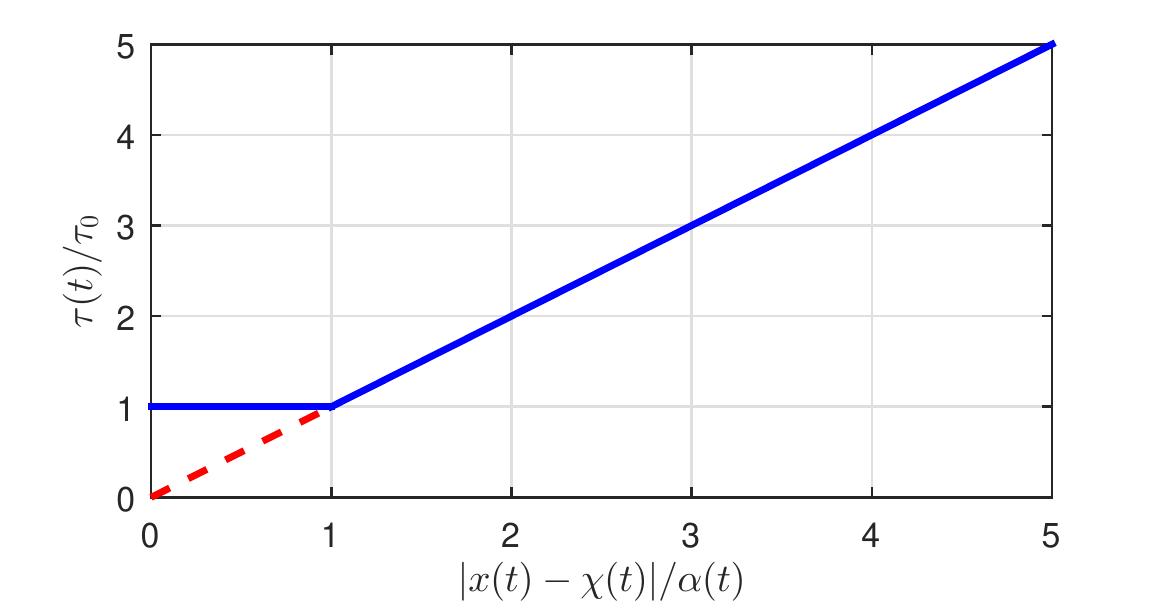}
\caption{ANDL time parameter ${\tau(t)= \tau(|x(t)-\chi(t)|)}$.}
\label{fig:Tau}
\end{figure}

\subsection{Resolution Parameter Calculation}

The configuration of the ANDL \cite{nikitin2014method} consists of a feedback mechanism that monitors the peakedness of the signal plus noise mixture and provides a time-dependent resolution parameter $\alpha(t)$ which ensures improvement in the quality of non-stationary signals under time-varying noise conditions. The idea is to pick an effective value of $\alpha(t)$ that allows the signal of interest to completely go through the nonlinear filter without any suppression and at the same time mitigate the impulsive noise, maximally. For implementation simplicity, we assume that SNR variations are slower relative to the OFDM symbol duration. Therefore, we can fix the resolution parameter for each OFDM symbol duration $\alpha(t){=}\alpha$ and allow it to change across symbols. The lower bound of the resolution parameter can be found based on difference signal $\left| {x(t) - \chi(t)} \right|$ in case of no impulsive noise. An estimate of the difference signal can be obtained by passing signal $s(t)+w(t)$ through a linear highpass filter. Let $z(t)$ be given by a differential equation for the 1st~order highpass filter with the time constant~$\tau_0$. Then, we have

\begin{equation} \label{eq:1st high pass}
  z(t)=\tau_0 \left[ \dot{s}(t)+\dot{w}(t)-\dot{z}(t) \right],
\end{equation}
Lemma~\ref{lem:lemma1} provides a lower bound for the choice of resolution parameter $\alpha$.

\newtheorem{lemma}{Lemma}
\begin{lemma}\label{lem:lemma1}
The efficient value of the resolution parameter ${\alpha_{\rm eff,\epsilon}}$ for $(1-\varepsilon)$ level distortionless filtering of the transmitted OFDM signal in thermal noise is $\erf^{-1}(1-\varepsilon)\sqrt{2}\sigma_z$, where $\sigma _z^2$ is the variance of $z(t)$ and $\varepsilon$ is a sufficiently small constant.
\end{lemma}
\begin{proof}

Since $s(t)$ and $w(t)$ are independent, for a sufficiently large $N$ it follows from the CLT that $z(t)$ is a complex Gaussian random variable with zero mean and variance $\sigma _z^2$.
From equations \eqref{eq:1st order CDL} and \eqref{eq:CDL tau},
the ANDL preserves its linear behavior for~$|z(t)|\le\alpha$.
Therefore, for $(1-\varepsilon)$ distortionless filtering of the transmitted OFDM signal in thermal noise, we require that
\begin{equation}
\Pr {\mkern 1mu} (\left| {z(t)} \right|{\mkern 1mu}  > {\mkern 1mu} \alpha ){\mkern 1mu} \le\varepsilon\ll 1.
\end{equation}
Since $z(t)$ is Gaussian, we have
\begin{equation} \label{eq:erf}
\Pr (\left| {z(t)} \right| > \alpha ) = 1 - \erf\left(\frac{\alpha }{\sigma_z\sqrt{2}}\right) \le\varepsilon,
\end{equation}
where $\erf(.)$ is the error function.
Solving equation \eqref{eq:erf} with respect to $\alpha$, we obtain
\begin{equation} \label{eq:erfinv}
\alpha_{\rm eff,\epsilon} \ge \erf^{-1}(1-\varepsilon)\sqrt{2} \sigma_z,
\end{equation}

\end{proof}

In practice, a choice of~${\varepsilon=4.68\times 10^{-3}}$ leads to ${\alpha\ge 2\sqrt{2}\, \sigma_z}$, i.e., ${\alpha_{\rm eff}=2\sqrt{2}\, \sigma_z}$ and we use sample variance instead of statistical variance $\sigma _z^2$ as it can be computed online and can track possible nonstationary behavior.


\section{Simulation results}

In this section, as a specific example we consider an OFDM-based NB-PLC in PRIME. Based on IEEE P1901.2 standard \cite{Standard} the sampling frequency has been chosen as $f_s=250$ kHz and the FFT size is $N=512$, i.e, the subcarrier spacing $∆f=488$ Hz. As carriers $N=86–-182$ are occupied for data transmission based on the PRIME model, the desired signal is located in the frequency range 42–-89 kHz \cite{Prime-G3}.
The system is investigated in a noise environment that is typical for NB-PLC and it consists of three components (1) thermal noise (with PSD decaying at rate of 30 dB per 1 MHz) (2) periodic cyclostationary exponentially decaying component with the repetition frequency at twice the AC line frequency and duration ranging from hundreds of microseconds to a few milliseconds, and (3) asynchronous random impulsive noise with normally distributed amplitudes captured by a poisson arrival process with parameter $\lambda$.

We use first order ANDL, with ${\tau _0} {=} 1/(2\pi {f_0})$ and corner frequency ${f_0} {=} 2 \times 89$ kHz, which is followed by a 2nd order linear filter with the time parameter $\tau=\tau_0$ and the quality factor $Q=1$. It is important to note that in the considered system model, the matched filter can take the role of the linear filter. When $\alpha  \to \infty $ this ANDL becomes a 3rd order Butterworth filter with cutoff frequency twice the highest frequency of the desired signal.
All simulations have been performed for BPSK modulation and the cyclostationary impulsive noise is simulated as a damped sinusoid based on equation \eqref{CS} and it lasts for $200 \mu s$ (one tenth of OFDM symbol). The asynchronous impulsive noise is added to the transmitted signal with different probability of impulsivity based on equation \eqref{AS} which lasts for $2 \mu s$. Since the cyclostationary noise is dominant in the NB-PLC, we set the power of this component three times higher than the asynchronous impulsive noise. We mimic the analog domain by oversampling the transmitted OFDM signal by factor 40 and downsampling after ANDL. In the following, BER of the OFDM system is used as the metric to evaluate the performance of ANDL in comparison with other conventional approaches such as linear filtering and blanking. Since, the noise is essentially stationary in the system, we can pick the effective $\alpha$ based on lemma \ref{lem:lemma1} for a fixed SNR leading to a classic ANDL implementation.
\begin{figure}[t]
\centering
\includegraphics[width=.5\textwidth,height=60mm]{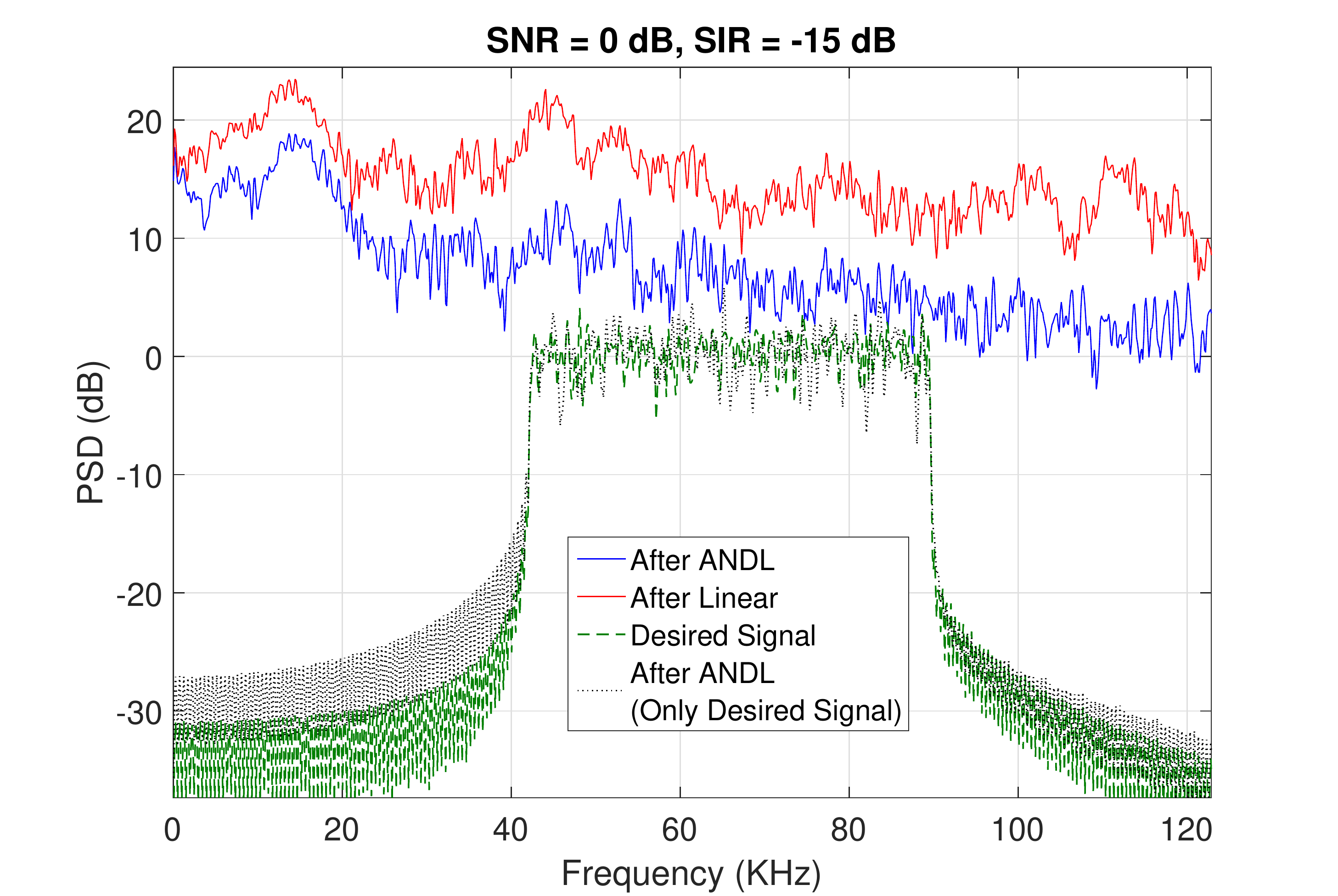}
\caption{Power Spectral Density.}
\label{fig:PSD}
\end{figure}
\begin{figure}[t]
\centering
\includegraphics[width=.5\textwidth,height=60mm]{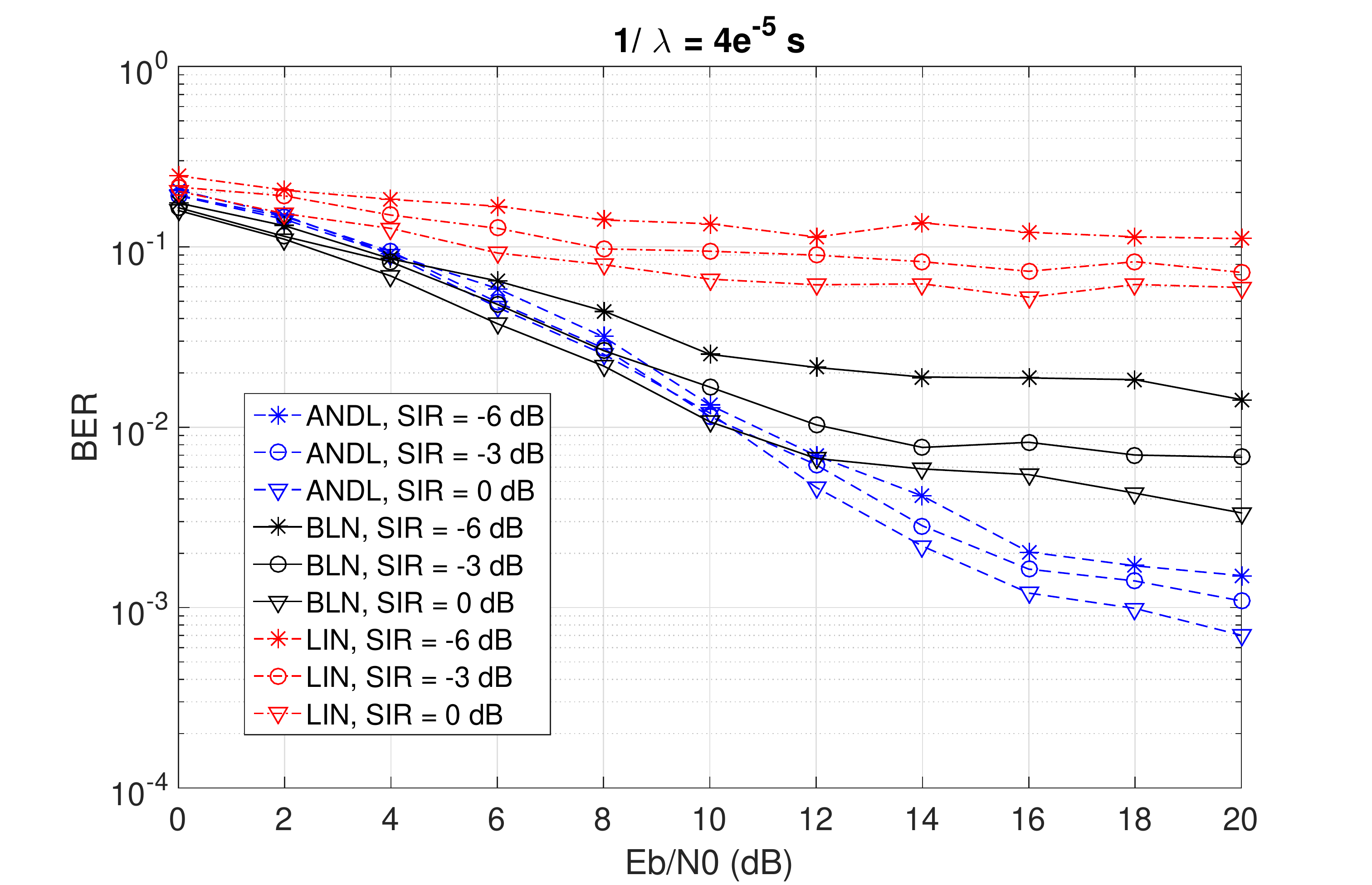}
\caption{BER versus SNR with fixed SIR.}
\label{fig:BER vs SNR, fixed SIR}
\end{figure}

Fig. \ref{fig:PSD} shows the power spectral density (PSD) for a given signal to thermal plus impulsive noise ratio (SINR) after impulsive noise mitigation filter. It is evident that we have significant impulsive noise suppression in passband with the ANDL compared to the suppression offered by a linear filter. This figure also shows that when there is no impulsive noise, the ANDL does not distort our desired signal in the passband. This disproportional effect of ANDL over the impulsive noise and desired signal in the passband results in significant SNR improvement at the receiver.

To demonstrate the robustness of the ANDL to different types of impulsive noise, we consider the case when both asynchronous and cyclostationary impulsive noise impact the signal simultaneously. The BER performance of proposed approach for different values of SIR versus SNR is shown in Fig. \ref{fig:BER vs SNR, fixed SIR}. We compare the ANDL performance with blanking and the optimal threshold for blanking is found based on an exhaustive numerical search. Fig. \ref{fig:BER vs SNR, fixed SIR} shows that the ANDL based reception results in better BER performance relative to blanking and linear filter especially in high SNR.
The BER performance of the system for a given SINR versus SNR is shown in Fig. \ref{fig:BER vs SNR, fixed SINR}. Since SINR is fixed, we have more impulsivity when thermal noise is low (i.e., high SNR region). Fig. \ref{fig:BER vs SNR, fixed SINR} shows that the performance of blanking and linear filter remains almost unchanged while the ANDL shows a significant improvement in high SNR region. This result highlights the effectiveness of the ANDL in severe impulsive noise environments.

The importance of choosing optimum resolution parameter $\alpha$ is shown in Fig. \ref{fig:alpha}. This figure shows the ANDL performance for different values of $\alpha$ for given amount of impulsive noise. We can see that the best performance is observed when $\alpha$ is selected based on lemma \ref{lem:lemma1}. As we deviate from this choice, the performance degradation is gradual and in many cases still superior to the linear filter performance (captured by setting $\alpha$ to a high value).

Finally, Fig. \ref{fig:AWGN} shows that ANDL, with proper selection of resolution parameter $\alpha$ based on lemma \ref{lem:lemma1} and sufficiently flat frequency response in passband, can be used as a general front end and operates as a linear filter when there is no impulsive noise. We achieve the theoretical AWGN performance indicating that our desired signal passes through the ANDL without any distortion as would be the case with a linear filter.

\begin{figure}[t]
\centering
\includegraphics[width=.5\textwidth,height=60mm]{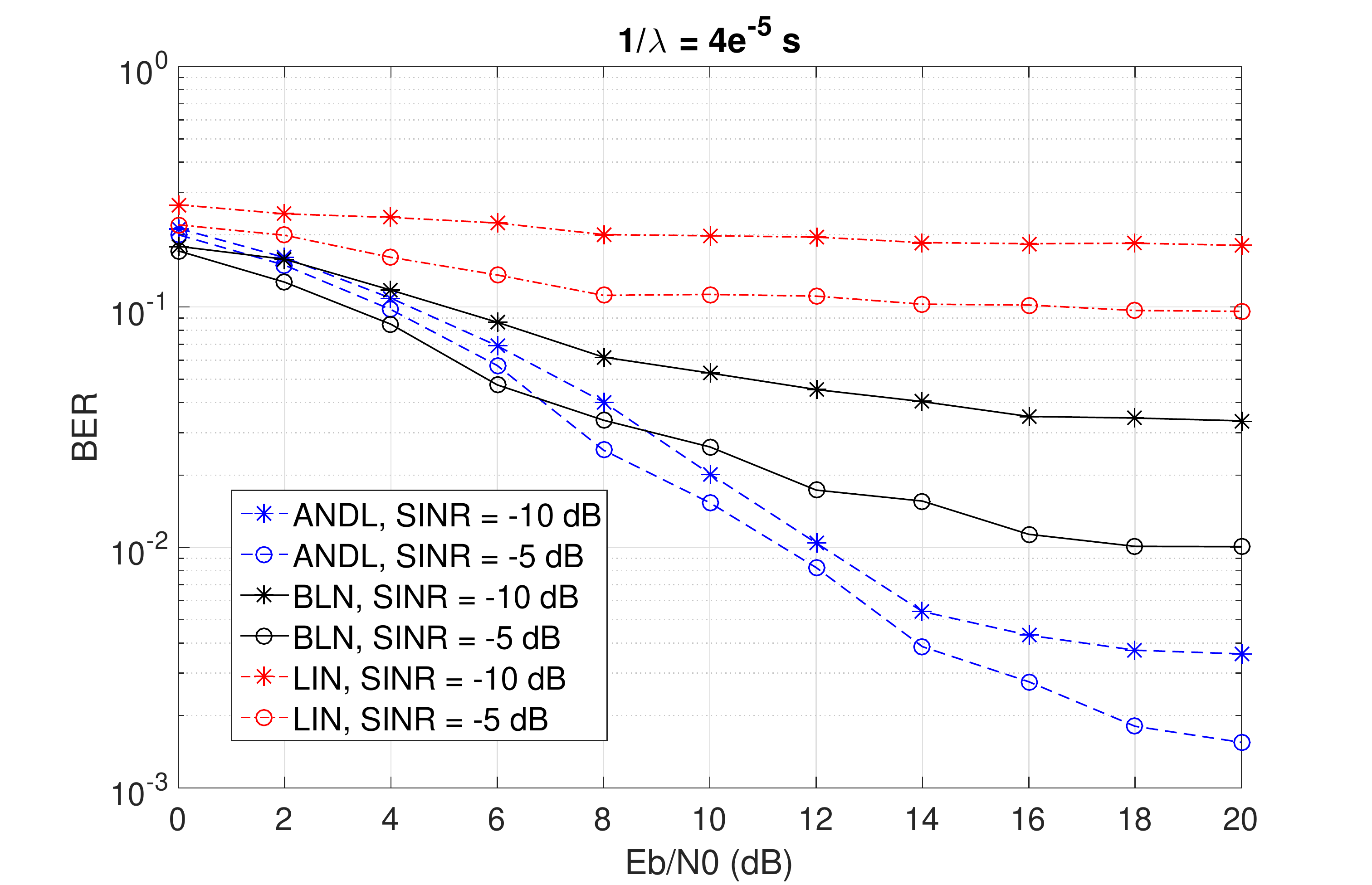}
\caption{BER versus SNR with fixed SINR.}
\label{fig:BER vs SNR, fixed SINR}
\end{figure}
\begin{figure}[t]
\centering
\includegraphics[width=.5\textwidth,height=60mm]{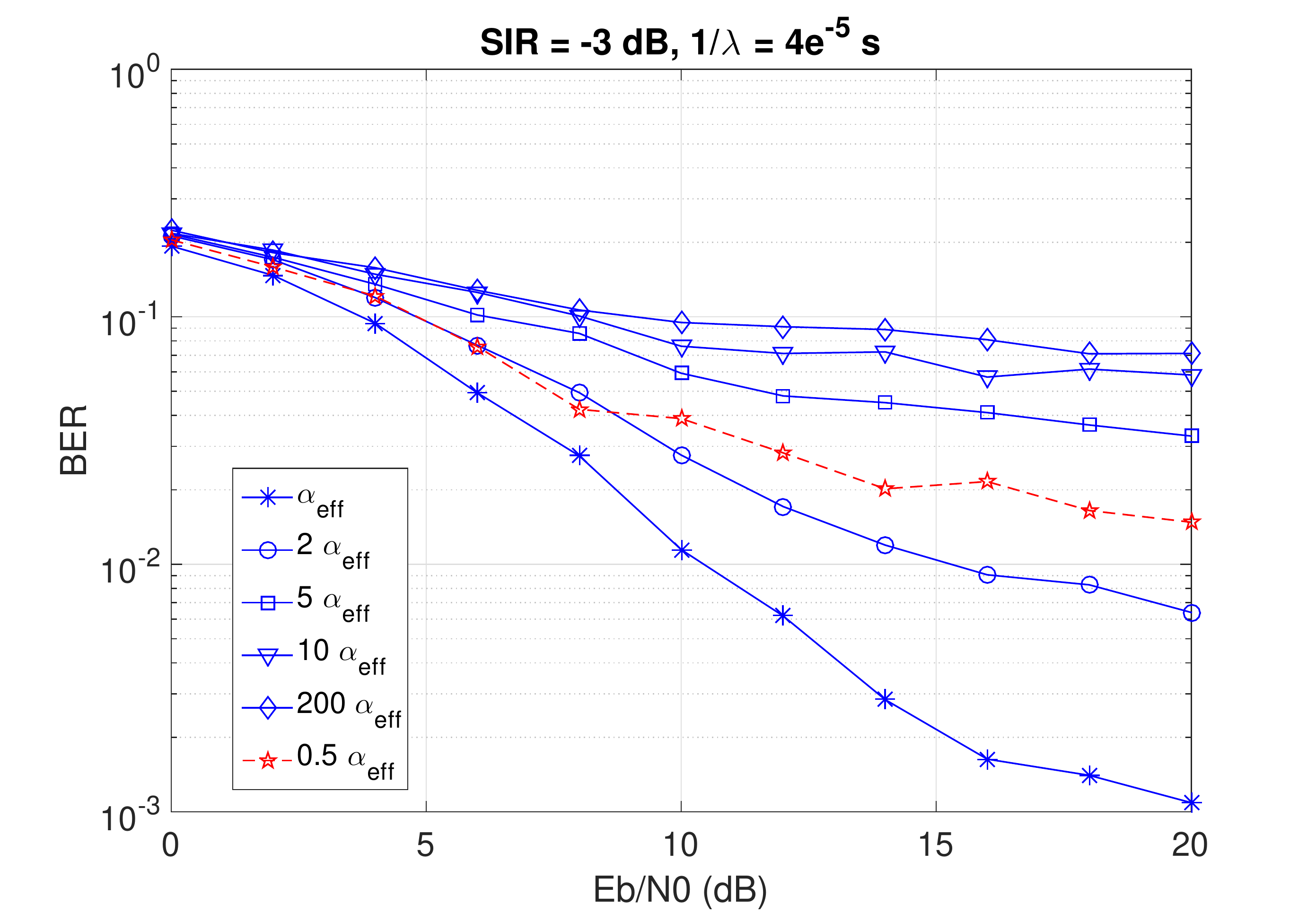}
\caption{Effect of resolution parameter on ANDL performance.}
\label{fig:alpha}
\end{figure}
\begin{figure}[t]
\centering
\includegraphics[width=.5\textwidth,height=60mm]{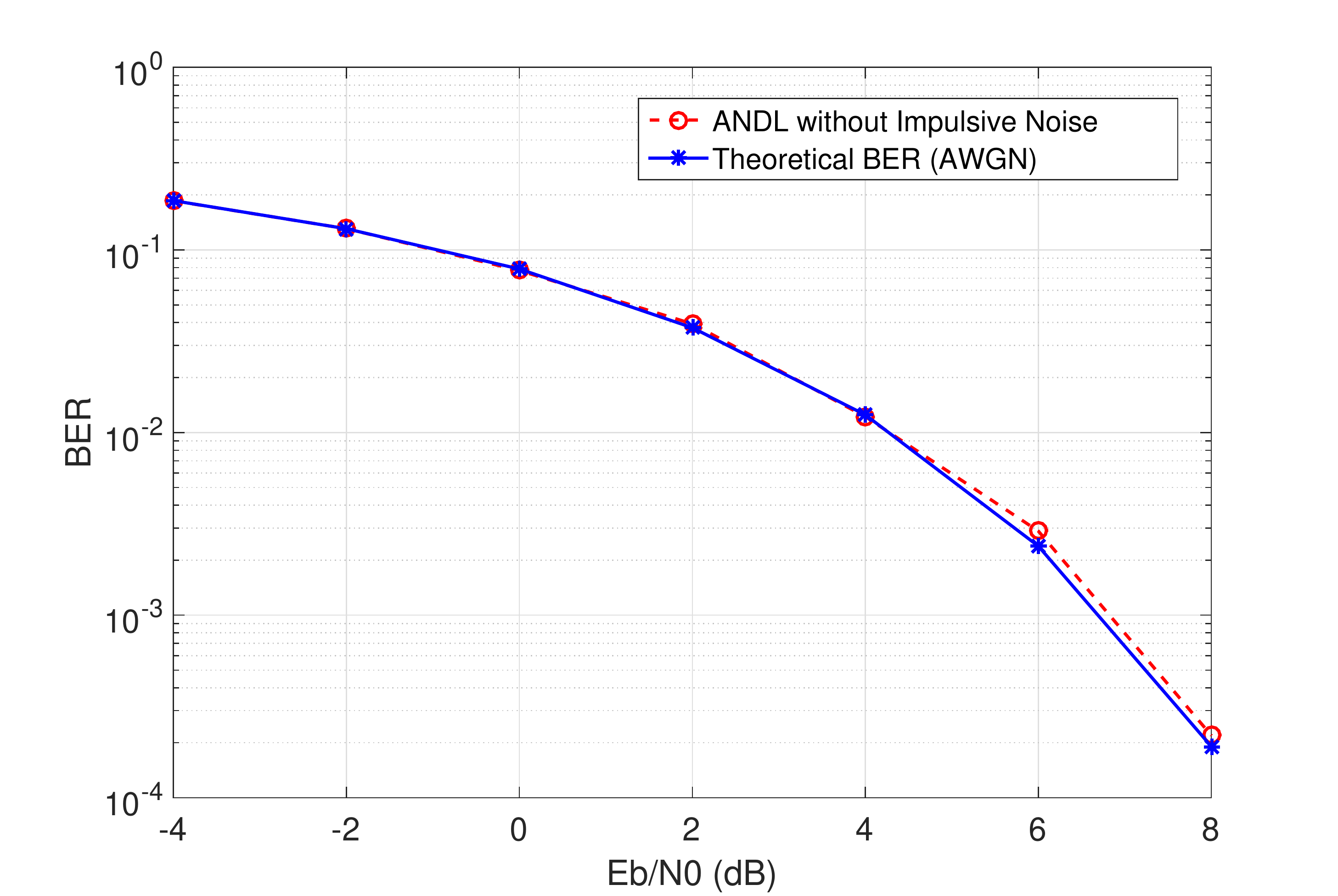}
\caption{ANDL performance without impulsive noise.}
\label{fig:AWGN}
\end{figure}


\section{Conclusion}

In this study, a blind adaptive analog nonlinear filter, referred to as Adaptive Nonlinear Differential Limiter (ANDL) is proposed to mitigate asynchronous and cyclostationary impulsive noises in OFDM-based PLC receiver. In addition, a practical method to find an effective value for the resolution parameter of ANDL is presented. We demonstrate the ability of ANDL to significantly reduce the PSD of impulsive noise in the signal passband without having prior knowledge of the statistical noise model or model parameters. The results show that ANDL can provide improvement in the overall signal quality ranging from distortionless behavior for low impulsive noise conditions to significant improvement in BER performance in the presence of strong impulsive component. It also has been shown that the performance of ANDL can be enhanced by careful selection of resolution parameter. It is important to note that ANDL can be deployed either as a stand-alone low-cost real-time solution for impulsive noise mitigation, or combined with other interference reduction techniques.

\bibliographystyle{IEEEtran}

\bibliography{IEEEabrv,Reference}

\end{document}